\documentclass[11pt,reqno,letterpaper]{amsart}
\usepackage{color}
\usepackage[colorlinks=true, allcolors=blue,backref=page]{hyperref}
\usepackage{amsmath, amssymb, amsthm}
\usepackage{mathrsfs}
\usepackage{mathtools}
\usepackage[noabbrev,capitalize,nameinlink]{cleveref}
\crefname{equation}{}{}
\usepackage{fullpage}
\usepackage[noadjust]{cite}
\usepackage{graphics}
\usepackage{pifont}
\usepackage{tikz}
\usepackage{bbm}
\usepackage[T1]{fontenc}

\usetikzlibrary{arrows.meta}

\usepackage{environ}
\usepackage{framed}
\usepackage{url}
\usepackage[linesnumbered,ruled,vlined]{algorithm2e}
\usepackage[noend]{algpseudocode}
\usepackage[labelfont=bf]{caption}
\usepackage{cite}
\usepackage{framed}
\usepackage[framemethod=tikz]{mdframed}
\usepackage{appendix}
\usepackage{graphicx}
\usepackage[textsize=tiny]{todonotes}
\usepackage{tcolorbox}
\usepackage{enumerate}
\allowdisplaybreaks[1]
\usepackage{enumerate}
\usepackage{stmaryrd}
\usepackage[margin=1in]{geometry}

\usepackage[shortlabels]{enumitem}
\crefformat{enumi}{#2#1#3}
\crefrangeformat{enumi}{#3#1#4 to~#5#2#6}
\crefmultiformat{enumi}{#2#1#3}
{ and~#2#1#3}{, #2#1#3}{ and~#2#1#3}

\DeclareSymbolFont{symbolsC}{U}{pxsyc}{m}{n}
\SetSymbolFont{symbolsC}{bold}{U}{pxsyc}{bx}{n}
\DeclareFontSubstitution{U}{pxsyc}{m}{n}
\DeclareMathSymbol{\medcircle}{\mathbin}{symbolsC}{7}

\crefname{algocf}{Algorithm}{Algorithms}

\crefname{equation}{}{} 
\AtBeginEnvironment{appendices}{\crefalias{section}{appendix}} 

\usepackage[color,final]{showkeys} 

\colorlet{refkey}{orange!20}
\colorlet{labelkey}{blue!30}

\crefname{algocf}{Algorithm}{Algorithms}

\numberwithin{equation}{section}
\newtheorem{theorem}{Theorem}[section]
\newtheorem{proposition}[theorem]{Proposition}

\newtheorem{claim}[theorem]{Claim}

\crefname{claim}{Claim}{Claims}

\newtheorem*{question*}{Question}

\theoremstyle{definition}
\newtheorem{definition}[theorem]{Definition}

\newtheorem*{definition*}{Definition}

\theoremstyle{remark}
\newtheorem*{remark}{Remark}

\newcommand{\mb}{\mathbb}

\newcommand{\mc}{\mathcal}

\newcommand{\ol}{\overline}
\newcommand{\on}{\operatorname}

\newcommand{\eps}{\varepsilon}

\let\originalleft\left
\let\originalright\right
\renewcommand{\left}{\mathopen{}\mathclose\bgroup\originalleft}
\renewcommand{\right}{\aftergroup\egroup\originalright}

\allowdisplaybreaks

\newcommand{\ignore}[1]{}

\title{Rapid mixing of the down-Up walk on matchings of a fixed size}
\author[A1]{Vishesh Jain}
\address{Department of Mathematics, Statistics, and Computer Science, University of Illinois Chicago, Chicago, IL, 60607 USA}
\email{visheshj@uic.edu}

\author[A2]{Clayton Mizgerd}
\address{Department of Mathematics, Statistics, and Computer Science, University of Illinois Chicago, Chicago, IL, 60607 USA}
\email{cmizge2@uic.edu}

\newcommand{\R}{\mathbb{R}}

\newcommand{\paths}{\operatorname{paths}}
\renewcommand{\P}{\mathcal{P}_\delta}

\newcommand{\E}{\mathbb{E}}

\begin{document}

\begin{abstract}
    Let $G = (V,E)$ be a graph on $n$ vertices and let $m^*(G)$ denote the size of a maximum matching in $G$. We show that for any $\delta > 0$ and for any $1 \leq k \leq (1-\delta)m^*(G)$, the down-up walk on matchings of size $k$ in $G$ mixes in time polynomial in $n$. Previously, polynomial mixing was not known even for graphs with maximum degree $\Delta$, and our result makes progress on a conjecture of Jain, Perkins, Sah, and Sawhney [STOC, 2022] that the down-up walk mixes in optimal time $O_{\Delta,\delta}(n\log{n})$. 

    In contrast with recent works analyzing mixing of down-up walks in various settings using the spectral independence framework, we bound the spectral gap by constructing and analyzing a suitable multi-commodity flow. In fact, we present constructions demonstrating the limitations of the spectral independence approach in our setting.

\end{abstract}

\maketitle



\section{Introduction}

Sampling and counting matchings in graphs is a central and well-studied problem. An early success in this direction is the classical algorithm of Kasteleyn for counting the number of perfect matchings in a planar graph~\cite{kasteleyn1967graph}. Starting with the foundational work of Valiant~\cite{V}, it was established that Kasteleyn's algorithm is exceptional in the sense that it is $\on{\#P-}$hard to (exactly) count the number of perfect matchings, even for restricted classes of input graphs such as bipartite graphs and graphs of bounded degree. In fact, perhaps quite surprisingly, the more general problem of counting matchings of a given size is $\on{\#P-}$hard, even restricted to the class of planar graphs~\cite{jerrum1987two}. 

Given the above hardness results, the best one can hope for is fully polynomial-time (possibly randomized) approximation schemes. In particular, in connection with fully polynomial-time randomized approximation schemes (FPRAS) for the number of matchings (possibly of a given size), as well as being an important problem in its own right, much work has been devoted to the problem of approximately sampling from various distributions on matchings of a graph. The celebrated work of Jerrum and Sinclair~\cite{JS} showed that for the monomer-dimer model at activity $\lambda$ (i.e.~the distribution on matchings where the probability of a matching $M$ is proportional to $\lambda^{|M|}$; $\lambda$ is known as the activity), the Glauber dynamics mixes in time polynomial in $n$ and $\lambda$. For graphs $G = (V,E)$ of bounded degree and $\lambda = O(1)$, the optimal mixing time $O(|E|\log{n})$ was obtained by Chen, Liu, and Vigoda~\cite{CLV}. 

By combining with a rejection sampling procedure, both of these works give polynomial time algorithms to approximately sample from the uniform distribution on matchings of size $k \leq (1-\delta)m^*(G)$ for any fixed $\delta > 0$, where $m^*(G)$ denotes the matching number of $G$ i.e.~the size of a largest matching in $G$; approximately sampling from the uniform distribution on perfect matchings of a graph remains a major open problem, although in the bipartite case, this was famously resolved by Jerrum, Sinclair, and Vigoda \cite{jerrum2004polynomial}. For the class of bounded degree graphs, an algorithm with near-optimal run time was provided by a recent work of Jain, Perkins, Sah, and Sawhney \cite{jain2022approximate}; they gave an algorithm which, given a graph $G$ of maximum degree $\Delta$, an integer $1 \leq k \leq (1-\delta)m^*(G)$, and a parameter $\varepsilon > 0$, outputs a random matching $M$ of size $k$ in time $\tilde{O}_{\Delta,\delta} (n)$\footnote{$\tilde{O}$ hides polylogarithmic factors in $n$ and $1/\varepsilon$.} such that the total variation distance is less than $\eps$ between the distribution on $M$ and the uniform distribution on $\mc{M}_k(G)$: the matchings in $G$ of size $k$.

Despite this progress, the mixing time of perhaps the simplest random walk on $\mc{M}_k(G)$ -- the so-called down-up walk -- is not understood. By the down-up walk for matchings of size $k$, we refer to the following chain:
\begin{enumerate}
    \item Denote the state at time $t$ by $M_t \in \mc{M}_k$. 
    \item Choose $e \in M_t$ and $e' \in E$ uniformly at random.\footnote{Another convention is to choose $e'$ uniformly at random among those edges for which $M_t \cup \{e'\} \setminus \{e\} \in \mc{M}_k$; in our setting, this would only have the effect of leading to a constant factor speed-up in the mixing time.} 
    \item Let $M' := M_t \cup \{e'\} \setminus \{e\}$. If $M' \in \mc{M}_k$, then $M_{t+1} = M'$. Else, $M_{t+1} = M_{t}$. 
\end{enumerate}
It is clear that the down-up walk is reversible with respect to the uniform distribution on $\mc{M}_k$ so that whenever it is ergodic (this need not be the case; for instance, consider the uniform distribution on perfect matchings of an even cycle), it converges to the uniform distribution on $\mc{M}_k$. It is believed that for any fixed $\delta > 0$ the down-up walk on $\mc{M}_k$ mixes in polynomial time for all $1\leq k \leq (1-\delta)m^*(G)$.\footnote{Some restriction on the range of $k$ is needed since, as just mentioned, the down-up walk is not even ergodic in general.} In fact, it was conjectured by Jain, Perkins, Sah, and Sawhney~\cite[Conjecture~1.4]{jain2022approximate} that for graphs $G$ of maximum degree $\Delta$ and $1 \leq k \leq (1-\delta)m^*(G)$, the $\eps$-total-variation mixing time of the down-up walk on $\mc{M}_k(G)$ is $O_{\Delta,\delta}(n\log(n/\eps))$, which would be optimal up to the implicit constants. 

The main result of this note establishes that the down-up walk on $\mc{M}_k(G)$ mixes in polynomial time for all $1 \leq k \leq (1-\delta)m^*(G)$. While our mixing time is unfortunately not sharp enough to resolve the aforementioned conjecture from \cite{jain2022approximate}, our result has the benefit of being applicable to arbitrary graphs (as opposed to graphs of bounded degree). 

\begin{theorem}
\label{thm:main}
Let $\delta \in (0,1)$. For a graph $G = (V,E)$ on $n$ vertices and $m$ edges, and an integer $1 \leq k \leq (1-\delta)m^*(G)$, the down-up walk on matchings of size $k$ has $\eps$-mixing time $O(n^{4/\delta}m^4 k \log(1/\eps))$.
\end{theorem}
\begin{remark}
Restricted to the class of graphs of maximum degree $\Delta$, our proof gives the improved $\eps$-mixing time bound of $O_{\Delta, \delta}(n^6k\log(1/\eps))$ by \cref{eq:inverse-spectral-gap} and \cref{eqn:sg-mixing}. We leave it as a very interesting open problem whether the mixing time can be improved to $\tilde{O}_{\Delta, \delta}(n)$ in this case (as was conjectured in \cite{jain2022approximate}). 
\end{remark}

Our proof is based on bounding the spectral gap using a carefully constructed flow. It is natural to ask whether the powerful spectral independence framework (developed in \cite{ALO}) can be used to derive a similar result; in \cref{sec:obstacles}, we present examples showing that there are serious barriers to this, even for the class of bounded degree graphs. Roughly, the main point is that the condition $k \leq (1-\delta)m^*(G)$ is not closed under pinnings (even if we take pinnings at random); this is not the case for the parameter range of independent sets considered in \cite{JMPV} and is key to making the spectral independence approach amenable in their setting.   

\subsection{Related work} For the down-up walk, notice that even the case when $G$ is itself a matching is already interesting; in this case, the down-up walk coincides with the classical and well-studied Bernoulli-Laplace chain to sample from the uniform distribution on $\binom{[n]}{k}$ (e.g.~\cite{diaconis1987time, lee1998logarithmic}). As discussed earlier, there are polynomial time algorithms, based on the rapid mixing of Glauber dynamics for the monomer-dimer model, to approximately sample from the uniform distribution on $\mc{M}_k(G)$, $1 \leq k \leq (1-\delta)m_k^*(G)$; instead of combining rejection sampling with the Glauber dynamics, one may also combine rejection sampling with a local random walk to sample from the uniform distribution on the union of matchings of size $k$ and $k-1$ (the rapid mixing of this walk is shown in \cite{dagum1988polytopes}). For bipartite graphs \cite{jerrum2004polynomial} and planar graphs \cite{alimohammadi2021fractionally}, there are polynomial time algorithms to approximately sample from the uniform distribution on $\mc{M}_k(G)$ for all $1 \leq k \leq m^*(G)$. 

Perhaps most relevant to this note is recent work of Jain, Michelen, Pham, and Vuong \cite{JMPV} which established optimal mixing of the down-up walk on \emph{independent sets} of a given size $1 \leq k \leq (1-\delta)\alpha_c(\Delta)n$, for the class of $n$ vertex graphs $G$ with maximum degree $\Delta$ (using the spectral independence framework). Here, $\alpha_c(\Delta)$ is a function such that the problem of (approximately) sampling independent sets of size $k > \alpha_c(\Delta)n$ on $n$ vertex graphs with maximum degree $\Delta$ is computationally intractable, unless $\on{NP} =\on{RP}$; this was shown by Davies and Perkins \cite{DP21}. In the same paper, Davies and Perkins showed that by combining the rapid mixing of the Glauber dynamics for the hard-core model in the tree uniqueness regime (\cite{ALO, CLV}) with a rejection sampling step, one can obtain a polynomial time algorithm to approximately sample from the uniform distribution on independent sets of size $k$ provided that $1 \leq k \leq (1-\delta)\alpha_c(\Delta)n$; this is entirely analogous to how \cite{JS, CLV} imply polynomial time approximate samplers for the uniform distribution on $\mc{M}_k(G)$ for $1 \leq k \leq (1-\delta)m^*(G)$. Davies and Perkins conjectured \cite[Conjecture~5]{DP21} that the down-up walk for independent sets mixes in polynomial time provided that $1 \leq k \leq (1-\delta)\alpha_c(\Delta)n$ and this was resolved (in a stronger form) by \cite{JMPV}; our work may be viewed as resolving the analog of the conjecture of Davies and Perkins for matchings.     

Finally, we remark that there is a large body of literature in probability concerned with the mixing of analogous walks for product(-like) domains with conservation laws (in our setting, the size of the matching is a conserved quantity); see, e.g.,~\cite{caputo2004spectral, filmus2022log} and the references therein. In our setting, the base measure (the natural choice is the monomer-dimer model at a suitable activity) is significantly more complicated and very far from being a product distribution, although we remark works are often able to exploit product structure and other symmetries to obtain rather precise results. 


\subsection{Organization} In \cref{sec:proof}, we present the proof of \cref{thm:main}. In \cref{sec:obstacles}, we discuss barriers to a potential spectral independence approach for proving \cref{thm:main}. In each (sub)section, we begin with an overview of the proof and some motivation.

\section{Proof of \cref{thm:main}}
\label{sec:proof}
\subsection{Preliminaries}

Let $P$ denote the transition matrix of an ergodic Markov chain on the finite state space $\Omega$, which is reversible with respect to the (unique) stationary distribution $\pi$.  Let $E(P) = \{ (x,y) \in \Omega \times \Omega : P(x,y) > 0 \}$ denote the ``edges'' of the transition matrix. Recall that the \emph{Dirichlet form} is defined for $f,g : \Omega \to \mb{R}$ by 
\[\mathcal{E}_P(f,g) := \frac{1}{2}\sum_{x,y \in \Omega}\pi(x)P(x,y)(f(x) - f(y))(g(x) - g(y)).\]
The \emph{spectral gap} $\alpha$ is defined to be the largest value such that for all $\varphi:\Omega \to \mb{R}$,
\[\alpha \on{Var}_{\pi}[\varphi] \leq \mc{E}_P(\varphi,\varphi).\]
The (total-variation) \emph{mixing time} is defined by
   \[ \tau_\mathrm{mix} = \max_{x \in \Omega} \, \min \{ t : d(P^t x, \pi)_\mathrm{TV} < 1/4 \},\]
   where $d(\cdot, \cdot)_{\on{TV}}$ denotes the total variation distance between probability distributions.\footnote{Note that the quantity $1/4$ is fairly arbitrary here.  Replacing $1/4$ with $\varepsilon$ increases the mixing time by at most a factor of $\log_2(\varepsilon^{-1})$.} 
The following relationship between the spectral gap and the mixing time is standard (see, e.g.~\cite{levin2017markov}):
\begin{equation}
\label{eqn:sg-mixing}
 \tau_{\text{mix}} \leq \alpha^{-1}\log\left(\frac{1}{\min_{x\in \Omega}\pi(x)}\right).
\end{equation}

In order to bound the spectral gap of the down-up walk, we will use the technology of multicommodity flows (\cite{S, diaconis1991geometric}). 

\begin{definition}
    \label{def:flow}
    Consider the undirected graph $H = (\Omega, E(P))$. For $x,y \in \Omega$, let $\mc{Q}_{xy}$ denote the set of all simple paths from $x$ to $y$ in $H$. Let $\mc{Q} = \bigcup_{x,y} \mc{Q}_{xy}$.  A \emph{flow} is a function $f : \mc{Q} \to \R_{\geq0}$ such that $\sum_{q \in \mc{Q}_{xy}} f(q) = \pi(x) \pi(y)$. Given a flow $f$, we define its \emph{cost} by
   \[\rho(f) = \max_{(x,y) \in E(P)} \frac{1}{\pi(x) P(x,y)} \sum_{q \ni (x,y)} f(q),\]
   and its \emph{length} by
   \[\ell(f) = \max_{q : f(q) > 0} |q|,\]
   where $|q|$ denotes the number of edges in the path $q$.
\end{definition}

It was shown in \cite{S, diaconis1991geometric} that any flow gives a lower bound on the spectral gap. 

\begin{theorem}
    \label{thm:spectral-gap-flow}
    Let $P$ be a reversible ergodic Markov chain and $f$ be a flow. Then the spectral gap $\alpha$ satisfies \[\alpha \geq 1/(\rho(f) \ell(f)).\]
\end{theorem}




For each $t \in E(P)$, let $\paths(t) = \{ q \ni t : f(q)>0 \}$.  A common tool (see, e.g.,~\cite{S}) for bounding $\rho(f)$ is a \emph{flow encoding}, which is a collection of maps $\eta_t : \paths(t) \to \Omega$ for all $t \in E(P)$.  If all the maps are $\operatorname{poly}(n)$-to-one and the measure $\pi$ is ``fairly tame'' (the uniform measure on $\Omega$ automatically satisfies this condition), then this gives an inverse polynomial bound on the spectral gap.




\subsection{Constructing a flow}\label{sec:CanonicalPaths}
Recall that $1 \leq k \leq (1-\delta)m^*(G)$ and $\Omega$ denotes the set of matchings in $G$ of size exactly $k$. Given two matchings $x,y \in \Omega$, our flow will be constructed by uniformly distributing the demand $\pi(x)\pi(y)$ over a collection of carefully constructed paths. Compared to the construction of a flow in \cite{JS}, we face two challenges: 
\begin{itemize}
\item First, since we are not working with perfect matchings (or matchings which are a constant additive size away from being perfect), using just one path to route all the flow for each pair of matchings in the natural fashion results in a flow with exponentially high cost. To get around this issue, we use the (standard) idea of distributing the flow uniformly among essentially all possible paths, as is done for the Bernoulli-Laplace model (see \cite{S}) and also for a random walk on the union of matchings of size $k$ and $k-1$ \cite{dagum1988polytopes}. 

\item Second -- and this is the main new ingredient in our construction -- our state space consists of matchings of a fixed size, whereas all previous walks and flow constructions (e.g.~\cite{JS, dagum1988polytopes}) required working with matchings of at least two adjacent sizes. In order to route flow along such paths while still incurring only polynomial cost, we divide pairs of matchings into a ``good'' set and a ``bad'' set depending on the combinatorial structure of the symmetric difference. For the good set, it is fairly simple to construct a flow, incorporating the above idea of distributing the flow uniformly among all possible paths. For a pair in the bad set $(x,y)$, we show that there is a nearby good pair $(\tilde{x},y)$, in the sense that $x$ can be transformed into $\tilde{x}$ using a short path. The fact that we can transform $x$ to a suitable $\tilde{x}$ with a short path is key to bounding the cost of the flow, and this is where we use that $k \leq (1-\delta)m^*(G)$.   
\end{itemize}

Let $x \oplus y$ denote the symmetric difference $(x \setminus y) \cup (y \setminus x)$.  Since $x$ and $y$ are each matchings and so have maximum degree 1, $x \oplus y$ is a disjoint union of paths and even-length cycles. For the sake of analysis, place arbitrary total orders on the set of even length paths in $G$ and the set of even length cycles in $G$.  Associate to each cycle one arbitrary distinguished vertex and to each odd-length path one arbitrary distinguished endpoint.  These will all remain fixed for the remainder of the paper.

We partition $\Omega^2 := \Omega \times \Omega$ into $(\Omega^2)_g \cup (\Omega^2)_b$ where
\[ (\Omega^2)_b = \{(x,y) : x \oplus y \text{ contains a cycle and no odd-length paths} \}, \]\[(\Omega^2)_g = \Omega^2 \setminus (\Omega^2)_b. \]
We will first describe the collection of paths between the ``good pairs'' $(\Omega^2)_g$. Later, we will leverage this collection on paths along with an additional idea to obtain a suitable collection of paths between the ``bad pairs'' $(\Omega^2)_b$.\\ 

\paragraph{\bf Good pairs} Let $(x,y) \in (\Omega^2)_g$.  The symmetric difference $x \oplus y$ consists of even length paths, even length cycles, odd length paths with more edges in $y$ (which are necessarily $x$-augmenting paths), and odd length paths with more edges in $x$ (which are necessarily $y$-augmenting paths).  We have an induced ordering on the even paths and the cycles from our total order.  Since $|x| = |y|$, $x \oplus y$ contains the same number of $x$-augmenting and $y$-augmenting paths; suppose there are $2j$ total odd-length paths.  Let $\sigma_x,\sigma_y$ be permutations of the sets of $x$-augmenting, $y$-augmenting paths respectively. For each such choice of $(\sigma_x, \sigma_y)$, we construct a path as follows from $x$ to $y$ in $\Omega$. Before proceeding to the formal details, let us briefly describe the procedure: we first change $x$ to $y$ along all even paths. We then change $x$ to $y$ along the first $x$-augmenting path $\sigma_x(1)$; at the end of such a path, there is an additional $y$-edge to be added, which gives us the necessary room to switch from $x$ to $y$ along all cycles, while still remaining in $\Omega$. At the end of the cycle processing stage, there is an additional $y$-edge to be added; we pair this up with switching $x$ to $y$ along the first $y$-augmenting path $\sigma_y(1)$. Finally, we switch $x$ to $y$ along pairs of $x$-augmenting and $y$-augmenting paths $\sigma_x(i), \sigma_y(i)$ in the natural fashion.  

Formally, set $M_0 = x$ and proceed as follows:
\begin{enumerate}
    \item Process all even length paths in order.  To process an even path, enumerate the edges $e_1,e_2,\ldots,e_{2\ell}$ such that $e_i \cap e_{i+1} \ne \emptyset$ and $e_1 \in y$.  This places all odd edges in $y$ and the evens in $x$.  Suppose $t$ steps have been taken.  First make the transition $M_{t+1} = M_t \cup e_1 \setminus e_2$,\footnote{As a slight abuse of notation, when $m$ is a matching, we will write $m \cup e$ to mean $m \cup \{e\}$ and similarly $m \setminus e$ instead of $m \setminus \{e\}$.} then $M_{t+2} = M_{t+1} \cup e_3 \setminus e_4$, and continue until $M_{t+\ell} = M_{t+\ell-1} \cup e_{2\ell-1} \setminus e_{2\ell}$.  After processing all even paths, if we have reached $y$, terminate.
    \item Process the first $x$-augmenting path $p = \sigma_x(1)$.  
    Let $e^* \in p$ be the edge incident to the distinguished endpoint.  Process $p \setminus e^*$ as an even path as in step (1), leaving only $e^*$ to be added.
    \item Process all cycles in order.  For a cycle $c$, let $e$ (respectively $e'$) be the edge in $c \cap x$ (respectively $c \cap y$) incident to the distinguished vertex of $c$.  First, let $M_{t+1} = M_t \cup e^* \setminus e$ to complete the previous path and puncture the cycle.  Now, process $c \setminus \{e,e'\}$ as an even path as in step (1).  Label $e^* := e'$ and process the next cycle in the same way.  At the end of this step, some $e^*$ will remain.
    \item Process the first $y$-augmenting path $p = \sigma_y(1)$.  Let $e \in p$ be the edge incident to the distinguished endpoint.  Begin with $M_{t+1} = M_t \cup e^* \setminus e$, then process $p \setminus e$ as an even path as in step (1).
    \item Process any remaining $x$-augmenting and $y$-augmenting paths in pairs $p = \sigma_x(i),p' = \sigma_y(i)$.  Let $e$ (respectively $e'$) denote the edges incident to the distinguished endpoints of $p$ (respectively $p'$). First process $p \backslash e$ as an even path, then exchange $M_{t+1} = M_t \cup e \setminus e'$, then process $p' \backslash e'$ as an even path.
\end{enumerate}

This defines a unique path from $x$ to $y$ for any two permutations $\sigma_x,\sigma_y$, and so gives $(j!)^2$ total paths $x \to y$.  We uniformly distribute the demand $\pi(x)\pi(y) = 1/|\Omega^2|$ by setting $f(q) = 1/(|\Omega|^2 (j!)^2)$ for each path $q$ thus defined.\\

\paragraph{\bf Bad pairs} Given $(x,y) \in (\Omega^2)_b$, we will route the flow through $(\Omega^2)_g$ by choosing some suitable $(\tilde{x},y) \in (\Omega^2)_g$ and adding a suitable prefix to all paths (as above) from $\tilde{x}$ to $y$.  Since $|x| \leq (1-\delta)m^*(G)$, $x$ has some augmenting path $p$ of length at most $2\delta^{-1}$.  This follows from the pigeonhole principle: for $M^*$ a maximum matching in $G$, $x \oplus M^*$ is a graph with at most $2m^*(G)$ non-isolated vertices and at least $\delta m^*(G)$ disjoint $x$-augmenting paths, so that there must be an $x$-augmenting path of length at most $2\delta^{-1}$. Consider now $x^+ := x \oplus p$; this is a matching of size $k+1$. We claim that there exists some $e \in x^+$ such that for $\tilde{x} := x^+ \setminus e$ satisfies $(\tilde{x},y) \in (\Omega^2)_g$. To see this, note that since $|x^+| > |y|$, $x^+ \oplus y$ contains a $y$-augmenting path $p'$. If $x^+ \subset p'$, then $x^+ \oplus y$ cannot contain a cycle, as $p'$ is an alternating path between edges in $x^+$ and in $y$, and $|x^+| = |y| + 1$, and so $y \subset p'$ as well.  Thus $x^+ \oplus y = p'$ is a single path and contains no cycles, so we may choose any $e \in x^+$ and $(\tilde{x}, y) \in (\Omega^2)_g$. Otherwise, by choosing any edge $e \in x^+ \setminus p'$, we guarantee $\tilde{x} \oplus y$ has odd-length paths (in particular, $p'$) and so $(\tilde{x},y) \in (\Omega^2)_g$. 

For every pair $(x,y) \in (\Omega^2)_b$, we make a fixed (but otherwise arbitrary) choice of $p$ (an $x$-augmenting path of length at most $2\delta^{-1}$) and $e$ as above. For a path $\tilde{q} \in \mc{Q}_{\tilde{x}y}$, we define $q \in \mc{Q}_{xy}$ as follows.
\begin{enumerate}
    \item Process $p$ as an $x$-augmenting path (see previous step 2), leaving some $e^*$ to be added.
    \item Make the exchange $M_{t+1} = M_t \cup e^* \setminus e$, arriving at $\tilde{x}$.
    \item Follow the path $\tilde{q}$.
\end{enumerate}

We assign $f(q) = f(\tilde{q})$ so that the same amount of flow is routed from $x$ to $y$ as from $\tilde{x}$ to $y$. We remark that choosing an augmenting path of length $O_\delta(1)$ is crucial for bounding the cost of the flow below. 



\subsection{Flow encoding}
For $t \in E(P)$, we now bound $f(t) := \sum_{q \ni t} f(q)$ using the method of flow encodings.  Recall that $\paths(t) = \{ q \ni t : f(q)>0 \}$.  Fix some transition $t = (z,z') \in E(P)$.  We will partition $\paths(t)$ into three sets and bound the contribution to $f(t)$ from each of the three using a ``partial flow encoding''.  We have the ``good'' paths $\paths_g(t)$ consisting of paths $q \in \paths(t)$ whose endpoints are in $(\Omega^2)_g$.  Recall that the paths in $(\Omega^2)_b$ consist of two phases: the prefix from $x \to \tilde{x}$, and then a good path from $\tilde{x} \to y$.  Denote by $\paths_a(t)$ those paths which use the transition $t$ in the prefix $x \to \tilde{x}$, and by $\paths_b(t)$ those paths which use the transition $t$ in the path from $\tilde{x} \to y$.
We will frequently need the set of short paths in $G$
\[ \P := \Big\{ (v_1v_2\cdots v_{\ell}) : \{v_i,v_{i+1}\} \in E(G), \ \ell \leq 2/\delta \Big\}. \]
We will construct $\Omega \times \mc{P}_\delta$-valued functions $\eta_g, \eta_b, \eta_a$ on these subsets of $\paths(t)$.\\

\paragraph{\bf Construction of $\eta_g$} We first construct $\eta_g : \paths_g(t) \to \Omega \times \P$. Let $t = (z,z') \in E(P)$. For a path $q$, let $q^-,q^+$ be the endpoints.  Let $m = q^- \oplus q^+ \oplus (z \cup z')$. It is easily checked that $m$ is a matching of size $k-1$ and that $m'\oplus (z\cup z') = q^- \oplus q^+$ (the same construction is used in \cite{JS}). For consistency, we further map $m$ into an element of $\Omega\times \mc{P}_{\delta}$ using a fixed (but otherwise arbitrary) $m'$-augmenting path $p \in \mc{P}_{\delta}$.  The existence of a short $m$-augmenting path is guaranteed by the fact that $|m| < (1-\delta)m^*(G)$. Formally, we have
\begin{align*}
    \eta_g : \paths_g(t) & \to \Omega \times \P \\
    q & \mapsto (q^- \oplus q^+ \oplus (z \cup z') \oplus p, p).
\end{align*}
Note that given the image $(m \oplus p, p)$, we take $(m' \oplus p) \oplus p$ to recover $m$, which then recovers  $q^- \oplus q^+$ as before.  We can now reindex the sum
\[ \sum_{q \in \paths_g(t)} f(q) = \sum_{(m,p) \in \Omega \times \P} \sum_{q \in \eta_g^{-1}(m,p)} f(q). \]
The endpoints of every $q \in \eta_g^{-1}(m,p)$ have the same symmetric difference as noted above.  Let this symmetric difference have $2j$ odd-length paths.  By our construction of the flow, $f(q) = |\Omega|^{-2} (j!)^{-2}$ for all $q \in \eta_g^{-1}(m,p)$.  We now count how many paths use the transition $(z,z')$ based on which $G$-paths it is processing. This requires some case analysis, but ultimately, is based on blending the analysis of the flow encoding for the Glauber dynamics for the monomer-dimer model in \cite{JS} with the flow encoding for the Bernoulli-Laplace model \cite{S}.\\

{\bf Case I:} If $(z,z')$ is processing even length paths, then we can use our total order on even length paths to identify which parts of each cycle belong to $q^-$ and to $q^+$.  We know that we have not yet begun processing odd paths or cycles, so the parts in $z$ belong to $q^-$ and those outside $z$ belong to $q^+$.  We thus know $q^-$ and $q^+$ and so there are exactly $(j!)^2$ paths using $(z,z')$.


{\bf Case II:} If $(z,z')$ finishes processing an odd-length $G$-path and begins processing a cycle, then we know the odd $G$-path is the first such processed, and by the same reasoning as before, we can deduce the endpoints $q^-$ and $q^+$.  We also know $\sigma_x(1)$ is the path intersecting $z \oplus z'$.  The remainder of $\sigma_x$ and the entirety of $\sigma_y$ is free, so there are $(j!)^2/j$ paths using $(z,z')$.

{\bf Case III:} If $(z,z')$ is processing entirely cycles, then $z$ has a perfect matching on $j+1$ of the odd-length paths and the interior edges on $j-1$ of the odd-length paths.  One of these $G$-paths has already been augmented, and then there will be a path for every ordering of the remaining $G$-paths.  Thus there are $(j+1)j!(j-1)! = (1 + 1/j)(j!)^2$ paths.

{\bf Case IV:} If $(z,z')$ finishes a cycle and begins an odd-length path, then $z$ has a perfect matching on $j+1$ odd paths (one of which is being de-augmented in $(z,z')$) and the interior edges on $j-1$ odd paths.  The path touched by $z \oplus z'$ is $\sigma_y(1)$.  We must choose one of the remaining $j$ perfectly matched paths to be $\sigma_x(1)$, and then the remainder of $\sigma_x,\sigma_y$ are free on the sets of $j-1$ interior, perfect paths respectively.  There are thus $j((j-1)!)^2 = (j!)^2/j$ paths using $(z,z')$.

{\bf Case V:} If $(z,z')$ is augmenting an odd-length path, then of the other $2j-1$ paths, $z$ is perfect on $j$ and interior on $j-1$.  Suppose $2r$ paths have already been processed.  Then we may choose the already-augmented paths ($\binom{j}{r}$ choices), the already-de-augmented paths ($\binom{j-1}{r}$ choices), the order for each ($(r!)^2$ choices), and the order for the remaining augmentations and de-augmentations ($(j-r)!(j-1-r)!$ choices).  This gives $j!(j-1)!$ paths through $(z,z')$ that have already processed $r$ pairs of $G$-paths.  We now sum over $0 \leq r \leq j-1$ to get $j (j!)(j-1)! = (j!)^2$ total paths through $(z,z')$.

{\bf Case VI:} If $(z,z')$ is de-augmenting an odd-length path, then by the same logic but with $j-1$ perfect paths and $j$ interior paths, we again have $(j!)^2$ total paths.

{\bf Case VII:} If $(z,z')$ finishes augmenting one odd path and begins de-augmenting the next, then we know these $G$-paths occur adjacently in the path $q$.  Then by the same reasoning as Case V but with $j-1$ perfect paths and $j-1$ interior paths, we will get $(j!)^2/j$ total paths through $(z,z')$.\\

In all cases, we have at most $(1+1/j)(j!)^2 \leq 2(j!)^2$ paths in $\eta_g^{-1}(m,p)$ that use the transition $(z,z')$.  As each has weight $f(q) = |\Omega|^{-2} (j!)^{-2}$, this means the inner sum is at most $2/|\Omega|^2$ and so we can bound
\begin{equation}\label{eqn:paths-g}
    \sum_{q \in \paths_g(t)} f(q) = \sum_{(m,p) \in \Omega \times \P} \sum_{q \in \eta_g^{-1}(m,p)} f(q) \leq \sum_{(m,p) \in \Omega \times \P} \frac{2}{|\Omega|^2} = \frac{2|\P|}{|\Omega|}.
\end{equation}

\paragraph{\bf Construction of $\eta_b$} Recall that $\paths_b(t)$ are those paths that use the transition $t$ as part of following a ``good'' path from $\tilde{x}$ to $y$.  Thus we may instead choose the good path $q$ that is routed through $t$, and then count how many starting points $x$ could route through $q^-$ (the starting point of $q$) to get to $q^+$ (the ending point of $q$).  By the construction of our paths, this requires $x \oplus q^-$ to be a single short augmenting $G$-path in $\P$ together with a single edge in $E(G)$.  These together will uniquely determine the total path $x \to q^+$.  Thus each good path through $t$ is used in at most $|\P|\,|E(G)|$ bad paths, and the value of $f$ is unchanged by the prefix, so we may bound using \cref{eqn:paths-g}
\begin{equation}\label{eqn:paths-b}
    \sum_{q \in \paths_b(t)} f(q) \leq \sum_{\tilde{q} \in \paths_g(t)} |\P|\,|E(G)| f(\tilde{q}) \leq \frac{2|\P|^2|E(G)|}{|\Omega|}.
\end{equation}


\paragraph{\bf Construction of $\eta_a$} Finally, for $t = (z,z^-)$ and $q \in \paths_a(t)$, let $p$ be the $G$-path that is augmented during the prefix.  Then, define the function
\begin{align*}
    \eta_a : \paths_a(t) & \to \Omega \times \P \\
    q & \mapsto (q^+, p).
\end{align*}
Suppose $q \in \eta_a^{-1}(m,p)$ for some matching $m \in \Omega$ and $G$-path $p \in \P$ and let $t = (z,z^-)$.  Then we know that $q^+ = m$, and we know that 
$q^-$ consists of $z \setminus p$ and the interior alternating edges of $p$.  Thus all paths in $\eta_a^{-1}(m,p)$ have the same endpoints, and so their total flow is at most the net flow between those two points, which is $|\Omega|^{-2}$.  We can then calculate
\begin{equation}\label{eqn:paths-a}
    \sum_{q \in \paths_a(t)} f(q) = \sum_{(m,p) \in \Omega \times \P} \sum_{q \in \eta_a^{-1}(m,p)} f(q) \leq \sum_{(m,p) \in \Omega \times \P} \frac{1}{|\Omega|^2} = \frac{|\P|}{|\Omega|}.
\end{equation}

\paragraph{\bf Bounding the cost of the flow} Using \cref{eqn:paths-g}, \cref{eqn:paths-b}, and \cref{eqn:paths-a}, for any transition $t \in E(P)$,
\begin{equation}\label{eqn:f(t)-bound}
    \sum_{q \in \paths(t)} f(q) = \sum_{q \in \paths_g(t)}f(q) + \sum_{q \in \paths_b(t)}f(q) + \sum_{q \in \paths_a(t)}f(q) \leq \frac{3|\P| + 2|\P|^2|E(G)|}{|\Omega|}.
\end{equation}
Since $\pi(z) = 1/|\Omega|$ for all $z \in \Omega$ and $P(z,z') \geq 1/(k |E(G)|)$ (since the possible transitions from $z$ consist of removing one of $k$ edges and adding one of $|E(G)|$ edges), we get that
\begin{equation*}
    \rho(f) \leq |\Omega|\cdot k |E(G)| \cdot \frac{3|\P| + 2|\P|^2|E(G)|}{|\Omega|} \leq 3k|E(G)|^2 |\P|^2.
\end{equation*}
Finally, let $\Delta$ denote the maximum degree of $G$ and note that $|\P| \leq 2n \Delta^{2/\delta - 1}$ to get that
\begin{equation}
    \label{eqn:flow-cost}
    \rho(f) \leq 12k|E(G)|^2 \cdot n^2\Delta^{4/\delta - 2}. 
\end{equation}



\subsection{Rapid mixing}
We will use \cref{thm:spectral-gap-flow} to bound the spectral gap via the flow $f$ defined in \cref{sec:CanonicalPaths}. Note that the down-up walk is reversible with respect to the uniform distribution, aperiodic since $P(x,x) > 0$, and irreducible (for instance, by using the paths used in our flow $f$). Therefore, the assumptions of \cref{thm:spectral-gap-flow} are satisfied. To bound the maximum length $\ell(f)$ of any path used in our flow, note that by construction, any edge in $G$ is included in at most three exchanges. Hence, $\ell(f) \leq 3|E(G)|$. Combining this with \cref{eqn:flow-cost}, we see that the spectral gap $\alpha$ of the down-up walk satisfies
\begin{equation} \label{eq:inverse-spectral-gap} \alpha^{-1} \leq 36 k |E(G)|^3\cdot n^2\Delta^{4/\delta - 2}. \end{equation}
Finally, the mixing time bound in \cref{thm:main} follows from \cref{eqn:sg-mixing} by noting that $\log |\Omega| \leq \log 2^{|E(G)|} \leq |E(G)|$. 







\section{Barriers to the spectral independence approach}
\label{sec:obstacles}
For a distribution $\pi$ on $\binom{[n]}{k}$, we define the (signed) pairwise influence matrix $M_{\pi} \in \mb{R}^{n\times n}$ by
\[M_{\pi} (i,j) = \begin{cases} 0 &\text{ if } j=i, \\ \mb{P}_{\pi}[j \,|\, i] - \mb{P}_{\pi}[j \,|\, \ol{i}] &\text{ otherwise.}\end{cases},\]
where
$\mb{P}_{\pi}[i] = \mb{P}_{S \sim \pi}[i \in S]$ and $\mb{P}_{\pi}[\ol{i}] = \mb{P}_{S \sim \pi}[i \notin S]$.
We say that $\pi$ is $\eta$-spectrally independent (at link $\emptyset$) if $\lambda_{\max}(M_{\pi}) \leq \eta$ and that $\pi$ is $\eta$-$\ell_\infty$-independent (at link $\emptyset$) if $\max_{i \in [n]}\sum_{j=1}^{n}|M_{\pi}(i,j)| \leq \eta$. Note that the latter condition implies the former.

We begin by noting that for the class of bounded degree graphs, for $k$ bounded away from the matching number, the uniform distribution on matchings of size $k$ is $O(1)$-$\ell_\infty$-independent. 

\begin{proposition}
\label{prop:empty-link-independence}    
Let $G = (V,E)$ be a graph on $n$ vertices with maximum degree $\Delta$. Let $\delta > 0$ and for $1 \leq k \leq (1-\delta)m^*(G)$, let $\pi$ be the uniform distribution on matchings of $G$ of size $k$. Then $\pi$ is $O_{\delta,\Delta}(1)$-$\ell_\infty$-independent (at link $\emptyset$).  
\end{proposition}

\begin{proof}
For $k = o(n)$, this is implied by a coupling argument (e.g.~\cite{liu2021coupling}). For $k = \Omega(n)$, the proof follows from the same argument as in the proof of \cite[Theorem~8]{JMPV}: the differences are that we compare to the monomer-dimer model at activity $\lambda = O_{\delta,\Delta}(1)$ using \cite[Lemma~4.1]{jain2022approximate}, replace \cite[Theorem~9]{JMPV} by \cite[Theorem~2.10]{CLV}, and replace \cite[Theorem~15]{JMPV} by a suitable multivariate zero-free region for the matching polynomial (e.g.~\cite{chen2022spectral}). We omit further details. 
\end{proof}

Given \cref{prop:empty-link-independence}, one might hope to obtain an inverse polynomial bound on the spectral gap of the down-up walk (at least for the class of bounded degree graphs) using the powerful spectral independence framework as is done, for instance, in the case of the down-up walk on independent sets of a fixed size in \cite{JMPV}; we refer the reader to \cite{ALO} for an introduction to this framework.  In order to do this, we need to show that the distribution remains $O_{\delta,\Delta}(1)$-spectrally independent under any pinning.  In our situation, a pinning $\tau$ is a matching of size $\ell < k$.  We would then consider $\Omega_\tau = \{ m \in \Omega : \tau \subset m \}$ under the distribution induced by $\pi$ (uniform, in our case), and show $O_{\delta,\Delta}(1)$-spectral independence of this space.  We note that there is a more powerful ``average-case'' version of this argument, which (roughly) allows us to consider typical pinnings obtained by starting from some fixed matching of size $k$ and pinning a random subset of $k-\ell$ edges to be included (see \cite{anari2022optimal, anari2024universality}). We present barriers to this approach.

\begin{itemize}
\item We observe that such an approach cannot work for the down-up walk. Indeed, if it were to work, then one would also be able to show that the down-up walk has inverse polynomial spectral gap for the induced uniform distribution on size $\ell$ matchings obtained by starting with an arbitrary matching of size $k$ and pinning a uniform subset of $k-\ell$ edges to belong to the matching. However, as we discuss below, it is easy to construct an example where even for polynomially large $\ell$, with high probability, the down-up walk is not even ergodic (\cref{claim:ergodic-failure}). 

\item In the above example, the failure of ergodicity may be circumvented by using an $O(1)$-step down-up walk. However, it is still the case that proving mixing of the $O(1)$-step down-up walk using the (average) spectral independence framework necessitates proving mixing for the $O(1)$-step down-up walk for the aforementioned induced distributions on matchings of size $\ell$. We present a construction (\cref{claim:pinning-gap-zero}) showing that these induced distributions can correspond to the uniform distribution on size $\ell$ matchings in pretty arbitrary graphs with matching number $\ell(1+o(1))$; hence, there does not seem to be a way to use this machinery without basically showing that $O(1)$-step down-up walks mix rapidly for (almost) maximum matchings in arbitrary bounded-degree graphs, which is a major open problem.  
\end{itemize}


Our examples will follow the same general template. To set up some notation, given a graph $G = (V,E)$ and a pinning $\tau$ (a matching $\tau$ in $G$), define the \emph{residual graph} $G_\tau$ to be the induced subgraph
\[G_\tau = G[V(G_\tau)],\]
where
\[ V(G_\tau) = V(G) \setminus \bigcup_{e \in \tau} e.\]
Sampling from the uniform distribution on matchings in $G$ of size $k$, conditioned on pinning $\tau$ to be in the matching, is equivalent to finding a matching of size $k-|\tau|$ in $G_\tau$. 

We are now ready to construct our examples.  Fix some $0 < \delta < 1/5$ the desired gap from maximality, as in the statement of \cref{thm:main}. We define a graph $G = (V,E)$ where $|V| = n$ as follows: $G$ consists of $\delta n/2$ disjoint copies of $P_9$ (the path with $10$ vertices and $9$ edges) and an arbitrary graph $G'$ on the remaining $(1-5\delta)n$ vertices such that $G'$ has a perfect matching. Let $M$ be the matching given by taking the union of a perfect matching $M'$ in $G'$ with the interior alternating edges on each $P_9$; note that $|M| = n/2 - \delta n/2 = (1-\delta)m^*(G)$.  We will be considering pinning a uniform random subset of $M$ of a fixed size.




\begin{claim}
\label{claim:pinning-gap-zero}
    For a random pinning $\tau$ of size $(1 - \lambda)|M|$, the ratio 
    \[ \E\left[ 1 - \frac{|M|-|\tau|}{m^*(G_\tau)} \right] = O(\delta \lambda^4). \]
\end{claim}

The implication of this claim is that, while we started with the uniform distribution on matchings of size at most $(1-\delta)$ of the maximum matching in $G$, we now need to deal with the uniform distribution on matchings of size at least $(1-\delta \lambda^4)$ times the maximum matching in $G_{\tau}$. 

\begin{proof}
    Let $\tau$ be a random pinning of size $(1-\lambda)|M|$.  Let $X_{\tau}$ be the number of $P_9$s that $\tau$ does not intersect.  Then by linearity of expectation,
    \begin{equation}\label{eqn:slack}
        \E[X_{\tau}] = \frac{\delta n}{2} \Pr[\tau \text{ avoids a fixed }P_9] = O(\delta n \lambda^5).
    \end{equation}
    The key observation here is that once we have pinned any edge in $M \cap P_9$ for some copy of $P_9$, we have split $P_9$ into two even paths and are demanding a maximum matching on each of those. Hence, by construction, we see that $m^*(G_\tau) = |M| - |\tau| + X_{\tau}$.  
    We now compute
    \begin{align*}
        \E\left[ 1 - \frac{|M| - |\tau|}{m^*(G_\tau)} \right] & = 1 - \E\left[ \frac{|M| - |\tau|}{|M| - |\tau| + X_\tau} \right] 
         \leq 1 - \frac{|M| - |\tau|}{|M| - |\tau| + \E[X_\tau]} \\
        & \leq \frac{\E[X_\tau]}{\lambda |M|} = O(\delta \lambda^4),
    \end{align*}
where the first line uses Jensen's inequality and the second line uses \cref{eqn:slack}. \qedhere
\end{proof}



In the above construction, take $G'$ to be a disjoint union of $(1-5\delta)n/4$ copies of $C_4$ (the $4$-cycle) and accordingly, take $M'$ to be a union of perfect matchings on each $C_4$. 

\begin{claim}
\label{claim:ergodic-failure}
    For a random pinning $\tau$ of size $n/2 - n^{2/3}$, with high-probability, the down-up walk on matchings of size $|M|-\tau$ on the induced graph $G_{\tau}$ is not ergodic. 
\end{claim}
\begin{proof}
It suffices to show that for a random pinning $\tau$ of size $n/2 - n^{2/3}$, with high probability, (i) $\tau$ intersects every $P_9$, (ii) $\tau$ fails to intersect some $C_4$. 

For (i), by a union bound and direct computation, we get that 
\[\mb{P}[\tau \text{ avoids some }P_9\cap M] \leq \frac{\delta n}{2}\mb{P}[\tau \text{ avoids a fixed }P_9 \cap M] = O(\delta n^{-1/3}).\]

For (ii), we get that
\begin{align*}
    \mb{P}[\tau \text{ intersects all }C_4] &= \mb{P}[\tau^{c} \text{ contains no } C_4 \cap M] \\
    &\leq n^{O(1)} \mb{P}[\tau^c \text{ does not contain a fixed }C_4 \cap M]^{(1-5\delta)n/4} \\
    & \leq \exp (-\Theta(n^{1/3})),
\end{align*}
where the second follows by comparing probabilities between the independent model of density $\Theta(n^{-1/3})$ and the slice model and the last line follows by direct computation. The union bound now shows that with high probability, (i) and (ii) simultaneously hold. 
\end{proof}

\section*{Acknowledgements}
V.J.~is supported by NSF CAREER award
DMS-2237646. C.M.~is supported in part by NSF award ECCS-2217023. We thank Huy Tuan Pham for helpful discussions and anonymous reviewers for several useful suggestions. 

\bibliographystyle{amsplain0.bst}
\bibliography{main.bib}

\end{document}